\documentclass[12pt]{article}
\usepackage{color}
\usepackage{hyperref}
\usepackage{amssymb,amsmath,amsthm}
\usepackage{bbold}

\newtheorem{Thm}{Theorem}[section]
\newtheorem{theorem}[Thm]{Theorem}
\newtheorem{proposition}[Thm]{Proposition}
\newtheorem{corollary}[Thm]{Corollary}
\newtheorem{lemma}[Thm]{Lemma}

\newtheorem{remark}{Remark}[section]

\def\undertilde#1{\mathord{\vtop{\ialign{##\crcr
$\hfil\displaystyle{#1}\hfil$\crcr\noalign{\kern1.5pt\nointerlineskip}
$\hfil\tilde{}\hfil$\crcr\noalign{\kern1.5pt}}}}}

\title{Stability of Relativistic Quantum Electrodynamics in the Coulomb Gauge}

\author{Christian D.\ J\"akel\footnote{jaekel@ime.usp.br, Dept.~de Matem\'atica Aplicada, 
Univ. de S\~ao Paulo (USP), Brasil} \ 
and Walter F. Wreszinski\footnote{wreszins@gmail.com, 
Instituto de Fisica, Universidade de S\~ao Paulo (USP), Brazil}}        

\begin{document}
\maketitle
\begin{abstract}
We show that relativistic quantum electro\-dynamics in the Coulomb gauge satisfies the following bound, 
which establishes stability: let $H(\Lambda,V)$ denote the Hamiltonian of 
$QED_{1+3}$ on the three-dimensional torus of volume $V$ and with ultraviolet cutoff $\Lambda$.
Then there exists a constant $0<\mu(\Lambda,V)<\infty$
(the vacuum energy renormalization) such that the renormalized Hamiltonian is positive: 
	\[
		H_{ren}(\Lambda,V) \equiv H_{\Lambda,V}+\mu_{\Lambda, V}\cdot \mathbb{1} \ge 0 \; . 
	\]
\end{abstract}

\section{Introduction}

The proof of stability of non-relativistic matter interacting with a classical 
electromagnetic field is one of the 
crown jewels of mathematical physics, both from the point of view of physics and 
mathematics (see \cite[Chapters~1-7 \& 9]{LiS} for a comprehensive exposition and references). 
It accounts for a wide, enormously rich class of phenomena in quantum mechanics, which are 
of crucial importance to the macroscopic world and even to everyday life.

There are, however, various phenomena which 
require the quantization of the electromagnetic field, 
such as spontaneous emission \cite{JJS} and the black-body radiation, with its astoundingly perfect fit to the 
spectrum of the cosmic radiation background \cite{WilPen}. In addition, there is the well-known conceptual 
necessity to quantize the electromagnetic field \cite{BohrR}, an argument which does not extend to the 
gravitational field (as Dyson recently remarked~\cite{Dyson}).

Quantization of the electromagnetic field is, however, a well-known source of trouble. 
Its coupling to non-relativistic 
matter  has been studied extensively, see \cite[Chapter 11, and references given there]{LiS} and 
Spohn's treatise \cite{Spohn}, Part II, Chapters 19 and 20, and also the references he cites. 
One particular important 
step was taken by Lieb and Loss \cite{LL1} (see also \cite[pp.~314--315]{Spohn}), who established  
an upper bound $O(\Lambda^{12/7})$ to the ground state energy. 
As the latter disagrees 
with the result $O(\Lambda^{2})$ suggested by perturbation theory, the bound by Lieb and Loss
implies that perturbation theory can not converge: it simply would yield a 
wrong picture of the electron cloud (see also \cite{Li2} for an 
illuminating discussion).  

For non-relativistic matter interacting with the radiation field, 
the term $\vec{p} \cdot \vec{A}_{\Lambda}$ in the Hamiltonian, where $\vec{p}$ stands for 
the momentum operator of a particle, and $\vec{A}_{\Lambda}$ for the quantized electromagnetic vector 
potential field with ultraviolet cutoff $\Lambda$, \emph{seems to indicate} a lower bound to the Hamiltonian only of 
type $-c \Lambda$, with $c > 0$ proportional to the number of static nuclei. Fr\"{o}hlich \cite{Frohlich} has 
remarked that, while such a bound proves stability of matter if an ultraviolet cutoff is imposed on the theory,
the linear dependence on $\Lambda$ is disastrous, physically speaking. He raised the question whether such 
a catastrophe does indeed prevail, relating it to Landau's conjectures (the so-called Landau pole, see \cite{Landau}).
He also remarked, at the time (and in this respect the situation
has not changed since) that \emph{it was not known} (provided a mass renormalization $M_{\Lambda}$ and a chemical potential 
renormalization $\mu_{\Lambda}$ are chosen appropriately) whether a lower
bound on $H_{\Lambda}(M_{\Lambda},\mu_{\Lambda})$ can be found which is \emph{uniform} in $\Lambda$. 

Similar problems are expected
in the case of relativistic quantum electrodynamics in the Coulomb gauge, where the term $\vec{j} \cdot \vec{A}_{\Lambda}$
plays a role similar to the above mentioned term $\vec{p} \cdot \vec{A}_{\Lambda}$, where now $\vec{j}$ denotes
a (regularized) electron-positron current, but where, in addition, a charge renormalization is expected from perturbation
theory.
In this paper we show that for the Hamiltonian of relativistic quantum electrodynamics in the Coulomb gauge 
a suitable vacuum energy renormalization  can be found such that the renormalized Hamiltonian (with no mass and
no charge renormalization) is positive. In doing so, we provide a more complete picture of the electron cloud.
It affects both the effective interaction between the electrons, 
and their kinetic energy, yielding a picture of ``dressed'' electrons and positrons. 

\section{QED$_{1+3}$ on the Three-Torus}
\label{sec:2}

We will study quantum electrodynamics in the Coulomb gauge, with a  
cut-off Hamiltonian  
	\begin{equation}
		H(\Lambda,V) \doteq \underbrace{H^\circ_{ferm.}(V) + H^\circ_{bos.} (V)}_{ =: H_\circ (V)} 
		+ H_{int}(V,\Lambda) + {:}H_{Coul.}(V,\Lambda){:} 
		\label{(1)}
	\end{equation}
acting on a Hilbert space $\mathcal{H} \doteq \mathcal{H}_{ferm.} \otimes \mathcal{H}_{bos.}$  
that is the tensor 
product of an antisymmetric Fock space $ \mathcal{H}_{ferm.}$ and 
a symmetric Fock space $\mathcal{H}_{bos.} $. 
The one-particle space for both Fock spaces is $\ell^{2}(\Gamma_\kappa)$, with 
	\[
		\Gamma_{V} \doteq  \left\{ \tfrac{2 \pi \nu}{V^{1/3}}  \mid \nu \in \mathbb{Z}^{3} \right\} 
		\quad 
		\text{and}
		\quad 
		\Gamma_\kappa \doteq \Gamma_{V} \cap \Lambda \; . 
	\]
Here $\Lambda$ is a \emph{finite} set, symmetric 
under inversion about each coordinate plane containing the origin, and therefore invariant
under inversion through 
the origin. The number of sites in $\Lambda$ will be denoted by $| \Lambda | $.  
Since $| \Lambda | $ is finite,  there are no antisymmetric $n-$particle functions for $n> |\Lambda|$. 
Hence, Pauli's principle ensures that $ \mathcal{H}_{ferm.}$ is finite-dimensional; 
however, this argument does not hold for bosons, and $\mathcal{H}_{bos.}$ is, in fact,  
infinite-dimensional. 
 
The first two operators on the r.h.s.~in \eqref{(1)} denote the free 
massiv fermion Hamiltonian 
	\begin{equation}
		H^\circ_{ferm.}(V) = \sum_{p \in \Gamma_\kappa} \sum_{\ell\in \{1,2\}} 
		\underbrace{ \sqrt{p^{2} + m^{2}} }_{=: \omega (p)} \; 
		\Bigl( b_\ell^{*}(p) b_\ell(p) + d_\ell^{*}(p) d_\ell (p) \Bigr)
		\otimes \mathbb{1}
		\; , 
		\label{(4.1)}
	\end{equation}
and the free massless boson Hamiltonian
	\begin{equation}
		H^\circ_{bos.} (V) = \mathbb{1} \otimes \sum_{k \in \Gamma_\kappa} 
		\sum_{\jmath \in \{1,2\}} |k| \, a_\jmath^* (k) a_\jmath (k) \; , 
		\label{(4.2)}
	\end{equation}
respectively. The constant $m>0$ appearing in \eqref{(4.1)} is the electron mass. The photon (resp.~electron and 
positron) annihilation and creation operators $a_\jmath (k),a_\jmath^* (k)$ (resp.~$b_\ell (p)$, $b_\ell^{*}(p)$, 
$d_\ell(p)$, $d_\ell^{*}(p)$)  are normalized so that
	\[
		\bigl[ a_\jmath (k),a_{\jmath'}^* (k') \bigr] = \delta_{k,k^{'}} \delta_{ \jmath,\jmath' } \; , 
	\]
and
	\[		
		\bigl\{ b_\ell (p) , b_{\ell'} (p') \bigr\} = \delta_{p,p^{'}} \delta_{\ell ,\ell'} \; , 
		\qquad
		\bigl\{ d_\ell(p) , d_{\ell'}^*(p') \bigr\} = \delta_{p,p^{'}} \delta_{\ell ,\ell'} \; , 
	\]
where $\delta$ is the Kronecker delta. All other commutators (or anti-commutators) are zero. 
The  operators $H^\circ_{ferm.}(V)$ and $H^\circ_{bos.}(V)$ depend on $V$ (but not on~$\Lambda$).
In the sequel, however, we frequently deal with operators that depend on both~$V$  and $\Lambda$, and 
so it will be convenient to set $\kappa = (V, \Lambda)$.

\goodbreak
The \emph{photon-fermion interaction}, with periodic boundary conditions, is	
	\begin{equation}
		H_{int}(\kappa) = - \int_{V} {\rm d}^3 x \,  J_{\kappa} (x)
		 \cdot e A_\kappa(x)\; . 
		\label{(3)}
	\end{equation}
The dot $ \cdot $ denotes the scalar product of  three-vectors in Euclidean space. The 
\emph{electric current} $J_\kappa (x) = \bigl(J^{(1)}_\kappa (x), J^{(2)}_\kappa (x), J^{(3)}_\kappa (x) \bigr)$ 
appearing in \eqref{(3)} is a three-vector-valued operator on ${\cal H}_{ferm}$, whose components can be 
expressed, using the Pauli matrices,  in terms of the electron-positron 
fields  $\Psi_\kappa(x)$:   
	\begin{equation}
	\label{(3a)}
		J^{(i)}_\kappa (x) \doteq  \langle \Psi_\kappa^{\dag}(x) ,  \alpha_i  
		\Psi_\kappa(x) \rangle \; , 
		\qquad \text{where $\alpha_i = \begin{pmatrix} 0 & 
		\sigma_i  \\
			\sigma_i  & 0 \end{pmatrix}  $} 
	\end{equation}
is a $4\times 4$-matrix for each $i=1, 2,3$. 
The symbol $\langle \, . \,  ,  \, . \, \rangle $ denotes the 
 scalar product in $\mathbb{C}^{4}$. The star on an operator denotes adjoint and the dagger above in the definition of the
current does not change the four vectors, in conformance with the above reference to the scalar product in $\mathbb{C}^{4}$, see below.
 The \emph{electron-positron field} itself is given by
	\begin{align}
		\Psi_\kappa(x) &\doteq \frac{1}{V^{1/2}} \sum_{p \in \Gamma_\kappa}
		\sum_{\ell \in \{1, 2\} }  
		\Bigl( b_{\ell} (p) \, u_{\ell} (p) \, {\rm e}^{i p \cdot x} + d_{\ell}^* (p) \, v_{\ell}(p) \, {\rm e}^{-i p \cdot x} \Bigr) \; , 
	\label{(11)}
	\end{align}
        \begin{align}
                \Psi_\kappa^{\dag}(x) &\doteq \frac{1}{V^{1/2}} \sum_{p \in \Gamma_\kappa}
		\sum_{\ell \in \{1, 2\} }  
		\Bigl( b_{\ell}^* (p) \, u_{\ell} (p) \, {\rm e}^{-i p \cdot x} + d_{\ell} (p) \, v_{\ell}(p) \, {\rm e}^{i p \cdot x} \Bigr) \; , 
        \end{align} 
with four-vectors
	\begin{align*}
		& u_{1} (p) = \tfrac{1}{\sqrt{2 \omega (p)  
									(\omega (p)+m)}} \left( \begin{smallmatrix}
									1 \\
									0 \\
									p_{3}  \\
									p_{1}+ip_{2}
									\end{smallmatrix} \right) , 
	\quad
		& u_{2} (p) = 
									\tfrac{1}{\sqrt{2 \omega (p)  
									(\omega (p)+m)}}
									\left( \begin{smallmatrix}
										0 \\
										1 \\
									p_{1}-ip_{2}  \\
									-p_{3}  
							 		\end{smallmatrix} \right) , 
	\\
		& v_{1} (p) = \tfrac{1}{\sqrt{2  \omega (p)  
									(\omega (p)+m)}}
									\left( \begin{smallmatrix}
									-p_{1}+ip_{2} \\
									p_{3} \\
									0 \\
									1 
							 		\end{smallmatrix} \right) , 
	\quad    
		& v_{2} (p) = \tfrac{1}{\sqrt{2 \omega (p) 
									(\omega (p)+m)}}
									\left( \begin{smallmatrix}
									p_{3}  \\
									p_{1}+ip_{2} \\
									1 \\
									0 
									\end{smallmatrix} \right) . 
	\end{align*}  
The Fermi field $\Psi_\kappa(x)$ is a \emph{bounded} operator for each $x$, 
but this statement does not hold for   
the vector potential $A_\kappa (x)$ appearing in \eqref{(3)}. The latter is given by
	\begin{align}
		A_\kappa (x) & \doteq \frac{1}{V^{1/2}} \sum_{k \in \Gamma_\kappa \setminus \{ 0 \}}
		\sum_{\jmath  \in \{1,2\}} \frac{\left( a_{\jmath} (k) {\rm e}^{i k \cdot x} + a_{\jmath}^* (k) {\rm e}^{i k \cdot x} \right)}
		{\sqrt{2|k|}}  \; 
		\epsilon_{\jmath} (k)  \; . 
	\label{(13)}
	\end{align}
The 
\emph{polarisation vectors} can be chosen (see \cite{LL2}) such that for $k_i \ge 0$, $i = 1,2,3$, 
	\begin{align}
		\epsilon_{1}(k) & \doteq \frac{1}{\sqrt{k_{1}^{2}+k_{2}^{2}}} \begin{pmatrix} k_{2} \\  -k_{1} \\ 0 \end{pmatrix} , 
	\label{(14.3)}
	\qquad \epsilon_{2}(k)  \doteq \frac{k}{|k|} \times \epsilon_{1}(k) 	\; , 
	\end{align}
and $\epsilon_{1}(-k) = \epsilon_{2}(k)$, as well as $\epsilon_{2}(-k) = \epsilon_{1}(k)$. The question of the choice of
polarization vectors in qed is a nontrivial matter, see also \cite{LL4}, and the forthcoming (38). 
Note that the three vectors $\bigl(\epsilon_{1}(k), \epsilon_{2}(k), \frac{k}{|k|}\bigr)$ form an oriented orthonormal 
basis in $\mathbb{R}^3$. Hence, by \eqref{(13)}, 
	\[
		\nabla \cdot A_\kappa(x) = 0 \; , 
	\]
which characterises \emph{Coulomb gauge}. 
The final term in \eqref{(1)} represents the \emph{Coulomb interaction} 
	\begin{equation}
		{:}H_{Coul.}(\kappa){:}  
		\doteq e^{2}  \int_{V \times V} {\rm d}^3x  {\rm d}^3y \,   {\cal V}_\kappa(x-y) 
		\, {:} \Psi_\kappa^{*}(x) \Psi_\kappa(x)  
		\Psi_\kappa^{*}(y) \Psi_\kappa(y){:}  
		\otimes \mathbb{1} \; . 
		\label{(18)}
	\end{equation}
We note that, with the above notation, $H_{Coul.}(\kappa)$ denotes the instantaneous Coulomb interaction \emph{without} normal
ordering, \emph{i.e.}, without the Wick dots. Here ${\cal V}_\kappa (x)$ denotes a regularized Coulomb potential on the torus:
	\begin{equation}
		{\cal V}_\kappa(x) = \frac{1}{V} \sum_{k \in \Gamma_\kappa \setminus \{ 0 \}} 
 		\frac{{\rm e}^{i k \cdot x} }{|k|^{2}} \; , \qquad 
		\kappa = (\Lambda, V) \; .  
	\label{(19)}
	\end{equation}
In the limit $\lim_{V \to \infty} \lim_{\Lambda \to \infty}  {\cal V}_{(\Lambda, V)}$ 
tends to the Coulomb potential in the distributional sense.
Note that due to the exclusion of $k=0$ in \eqref{(13)}
${:} \, H_{int}(\kappa)  \, {:} = H_{int}(\kappa)$.

\goodbreak
It was proven in \cite{JLW}
that there exists a dense set of vectors $\mathcal{D}$ in $\mathcal{H}$ 
on which $H(\Lambda,V)$ is essentially self-adjoint. The closure of  
$H(\Lambda,V)$ has a purely discrete spectrum with finite multiplicity, it is bounded from below, 
and the eigenfunctions of $H(\Lambda,V)$ lie in $\mathcal{D}$. 

\section{A finite-dimensional Grassmann algebra}

In the sequel, $c_\ell^* (p) = b_\ell^* (p), d_\ell^* (p)$ will denote either an electron or a positron creation operator.  
To every vector $\Psi \in \mathcal{H}_{frem.}$, written in the form
	\begin{equation}
		\Psi = \sum_{n=0}^{\infty} \; \; \sum_{p_i \in \Gamma_\kappa}
		\sum_{\ell_i \in  \{ 1, 2\}} \; K^{(n)}_{ \ell_1, \ldots , \ell_n} (p_{1}, \ldots , p_{n}) \, 
		c_{\ell_1}^{*}(p_{1})  \cdots  c_{\ell_n}^{*}(p_{n})  \Omega_\circ \; , 
	\label{(22)}
	\end{equation}
we will now associate an element
	\[
		\mathbb{\Psi} (\mathbb{c}^*) = \sum_{n=0}^{\infty} \frac{1}{n!^{1/2}} \; \; \sum_{p_i \in \Gamma_\kappa}
		\sum_{\ell_i \in  \{ 1, 2\}} \; K^{(n)}_{ \ell_1, \ldots , \ell_n} (p_{1}, \ldots , p_{n}) \,  
		\mathbb{c}_{\ell_1}^*(p_1) \cdots \mathbb{c}_{\ell_n}^*(p_n)
	\]
of the Grassmann algebra ${\cal G}_{n}$ generated by anti-commuting 
symbols $\mathbb{c}_\ell (p)$ and $\mathbb{c}_{\ell'}^* (p')$:
	\begin{equation}
		\bigl\{ \mathbb{c}_\ell (p),\mathbb{c}_{\ell'}(p') \bigr\} 
		= \bigl\{ \mathbb{c}_{\ell}(p) ,\mathbb{c}^{*}_{\ell'}(p') \bigr\} 
		= \bigl\{ \mathbb{c}^*_{\ell}(p) , \mathbb{c}^*_{\ell'}(p') \bigr\} = 0 \; , 
	\label{(21)}
	\end{equation}
with $p, p' \in \Gamma_\kappa$ and $\ell, \ell' \in \{1,2\}$. The algebra ${\cal G}_{n}$ 
is of dimension $n = 2^{8 | \Lambda |}$; see \cite[p.~49]{Berezin}. The set of~$\mathbb{\Psi}$'s corresponding to 
the state vectors $\Psi \in \mathcal{H}_{frem.}$ is denoted by ${\cal L}$. 

\goodbreak
To further simplify the notation, we enumerate the momenta in $p_i \in \Gamma_\kappa$, 
and the set 
	\[
		\mathbb{c}_{2i+\ell -1}  \equiv \mathbb{c}_{\ell} (p_i) \; , 
		\qquad
		i= 1, 2, \ldots,   2^{8 | \Lambda |-1}  \; ,  \quad \ell \in \{1, 2 \} \; . 
	\]	
Next, one introduces symbols ${\rm d} \mathbb{c}_{1}, \ldots, {\rm d}\mathbb{c}_{n}$, subject to the commutation relations
	\begin{equation}
		\{ {\rm d}\mathbb{c}_{i},{\rm d}\mathbb{c}_{k} \}
		= \{\mathbb{c}_{k},{\rm d}\mathbb{c}_{i} \}
		= \{ {\rm d}\mathbb{c}_{i}^{*},{\rm d}\mathbb{c}_{k}^{*}\}
		= \{\mathbb{c}_{k}^{*},{\rm d}\mathbb{c}_{i}^{*}\} = 0 \, ,  
		\; \; 
		 i,k = 1,  \ldots,   2^{8 | \Lambda |} \; , 
	\label{(25.1)}
	\end{equation}
and defines single integrals
	\begin{equation}
		\int {\rm d}\mathbb{c}_{i} = 0 \; ,  \quad
		\int {\rm d}\mathbb{c}_{i}^{*} = 0 \; ,  \quad \int {\rm d}\mathbb{c}_{i} \; \mathbb{c}_{i}  = 1 
		\quad \text{and} \quad
		\int  {\rm d}\mathbb{c}_{i}^{*} \; \mathbb{c}_{i}^{*} = 1 \; . 
		\label{(25.2)}
	\end{equation}
Multiple integrals are understood as iterated integrals. Thus, \eqref{(25.1)} and \eqref{(25.2)} define the 
integral 
	\[
		\int {\rm d}\mathbb{c}_{n} \cdots {\rm d}\mathbb{c}_{1}\; f(\mathbb{c}) 
	\]
for all monomials and one can than extend the integrals to 
arbitrary elements by linearity. 

The Grassmann algebra  
${\cal G}_{n}$ has an involution \cite[p.~66]{Berezin} and, corresponding 
to the inner product $\langle \Psi_{1}, \Psi_{2} \rangle$ in $\mathcal{H}_{ferm.}$, there is an associated 
inner product in ${\cal L}$ 
	\begin{align}
		\langle \mathbb{\Psi}_{1}, \mathbb{\Psi}_{2} \rangle & 
		= \int \prod {\rm d}\mathbb{c} {\rm d}\mathbb{c}^{*}\, {\rm e}^{-\mathbb{c} \mathbb{c}^{*}}  \; 
		\mathbb{\Psi}_{1} \mathbb{\Psi}_{2} \; , 
		\label{(26)}
	\end{align}
where ${\rm e}^{-\mathbb{c} \mathbb{c}^{*}} \doteq \prod_{i=1}^{n} {\rm e}^{-\mathbb{c}_{i}\mathbb{c}_{i}^{*}}$
and $\prod {\rm d}\mathbb{c}{\rm d}\mathbb{c}^{*} \doteq \prod_{i=1}^{n} {\rm d}\mathbb{c}_{i} {\rm d} \mathbb{c}_{i}^{*} $. 

\begin{lemma}
The identity \eqref{(26)} converts ${\cal L}$ into a Hilbert space, which serves as a realization 
of the fermionic Fock space $\mathcal{H}_{ferm.}$ (\cite[Theorem 3.1, p.~83]{Berezin}). 
\end{lemma}

We next discuss how a bounded operator $A$ acting on $\mathcal{H}_{ferm.}$ is represented 
by an element $\mathbb{A} \in {\cal G}_{n}$, 
which naturally acts on ${\cal L}$. The connection is particularly simple, if $A$  is given in its \emph{normal form} 
	\begin{align}
		A & = \sum_{m,n} \sum_{p_i, p'_j \in \Gamma_\kappa}
		\sum_{\ell_i , \ell_j' \in  \{ 1, 2\}} \; 
		K^{(m, n)}_{\ell_1, \ldots , \ell_m, \ell'_1, \ldots , \ell'_n} 
		(p_{1}, \cdots ,p_{m}, p'_{1}, \cdots, p'_{n}) 
		\nonumber \\
		& \qquad \qquad \qquad \qquad \qquad  \qquad
		\times a_{\ell_1}^{*}(p_{1}) \cdots a_{\ell_m}^{*}(p_{m})
		a_{\ell_1'}(p'_{1}) \cdots a_{\ell_n'}(p'_{n}) \; , 
	\label{(29.1)}
	\end{align} 
as this allows us to represent $A$ by  (see \cite[p.~26]{Berezin}) 
	\begin{align}
			\mathbb{A} (\mathbb{c}^*, \mathbb{c}) & \doteq \sum_{m,n} \sum_{p_i, p'_j \in \Gamma_\kappa}
			\sum_{\ell_i , \ell_j' \in  \{ 1, 2\}} \; K^{(m, n)}_{\ell_1, \ldots , \ell_m, \ell'_1, \ldots , \ell'_n} 
			(p_{1}, \cdots ,p_{m}, p'_{1}, \cdots, p'_{n})
	\label{(29.2)}
			 \\
			& \qquad \qquad \qquad \qquad \qquad \qquad  \times \mathbb{c}_{\ell_1}^{*}(p_{1}) \cdots \mathbb{c}_{\ell_m}^{*}(p_{m}) 
			\mathbb{c}_{\ell_1'}(p'_{1})) \cdots \mathbb{c}_{\ell_n'}(p'_{n}) \; .   
			\nonumber
	\end{align}
In fact, if $\Psi= A \, \Phi $, then $\Psi$ 
corresponds to the element  (using a suggestive notation)
	\begin{equation*}
		\mathbb{\Psi}(\mathbb{c}^{*}) = \int \prod {\rm d} \mathbb{f}^{*} 
		{\rm d} \mathbb{f} \,  {\rm e}^{-(\mathbb{f}^{*}
		-\mathbb{c}^{*}) \mathbb{b}}  \; 
		\mathbb{A}(\mathbb{c}^{*},\mathbb{f}) 
		\mathbb{\Psi}(\mathbb{f}^{*})  \in  {\cal L}  \; ; 
	\end{equation*}
see \cite[Equ.~(3.68), p.~84]{Berezin}. 

\begin{remark}
It is important to notice that the correspondence between \eqref{(29.1)} 
and \eqref{(29.2)} only holds for operators $A$ equal to sums of normal forms, in which $A$ 
are (at most) linear in each of the $b_\ell (p)$, $d_\ell (p)$, $b^*_\ell (p)$ and $d^*_\ell (p)$, 
due to the Pauli exclusion principle. 
\end{remark}

\goodbreak 
Returning to the model introduced in Section \ref{sec:2}, we 
shall consider expressions of the form
	\[
		\langle \mathbb{\Omega} , \mathbb{H} (\Lambda,V) \mathbb{\Omega} \rangle \; , 
	\]
where $\mathbb{\Omega} \in {\cal L} \otimes \mathcal{H}_{bos.}$ is a vector of unit norm of the form 
	\begin{align}
		\mathbb{\Omega} & = \sum_{n,m} \frac{1}{\sqrt{ n! \, m!}} 
		\sum_{p_i , p_j'\in \Gamma_\kappa} 
		\sum_{\ell_i , \jmath_k \in \{  1,2 \} } \; 
		K^{(m, n)}_{\ell_1, \ldots , \ell_m, \jmath_1, \ldots , \jmath_n} 
			(p_{1}, \cdots ,p_{m}, p'_{1}, \cdots, p'_{n})
		\nonumber \\
		& \qquad \qquad \qquad \qquad \qquad 
		\times \mathbb{c}_{\ell_1}^{*}(p_{1}) 
		\cdots \mathbb{c}_{\ell_m}^{*}(p_{m}) \, a_{\jmath_1}^{*}(p'_{1}) 
		\cdots  a_{\jmath_n}^{*}(p'_{n})  \Omega^\circ_{bos.} \; . 
	\label{(30.2)}
	\end{align}
Here $\Omega^\circ_{bos.}$ denotes the vacuum (no-particle state) in $\mathcal{H}_{bos.}$, 
the sums over n and m are finite,  and 
	\[
		\sum_{m,n} \; \sum_{p_i , p_j'\in \Gamma_\kappa} 
		\sum_{\ell_i , \jmath_k \in \{  1,2 \} }
		 \; \bigl| K^{(m, n)}_{\ell_1, \ldots , \ell_m, \jmath_1, \ldots , \jmath_n} 
			(p_{1}, \cdots ,p_{m}, p'_{1}, \cdots, p'_{n}) \bigr|^{2}  = 1 \; .
	\]
The $a_{\jmath}^{*}(k)$'s appearing in \eqref{(30.2)} are the photon creation operators. 

\begin{proposition}
The set of vectors $\mathbb{\Omega}$  of the form 
\eqref{(30.2)} 
form a dense set in $C^{\infty}(  H_\circ (V) )$,  on which   
$H(\kappa)$ is essentially self-adjoint 
by  \cite[Theorem~3.1]{JLW}.  
\end{proposition}

\section{The first unitary transformation}

In this section, we apply a unitary transformation to the Hamiltonian 
	\begin{equation}
		\mathbb{H}(\kappa) = \mathbb{H}^\circ_{ferm.}(V) + \mathbb{H}_{Coul.}(\kappa)  
			+ \mathbb{H}_{\, int}(\kappa) + H^\circ_{bos.}(V) \; , 
		\label{(33.1)}
	\end{equation}
where both
	\begin{align*}
		\mathbb{H}^\circ_{ferm.}(V)  & = 
		\sum_{p \in \Gamma_\kappa} \sum_{\ell\in \{1, 2\} } \sqrt{|p|^{2} + m^{2}}
		 \Bigl( \mathbb{b}_\ell^* (p) \mathbb{b}_\ell (p)+\mathbb{d}_\ell^* (p)
		 \mathbb{d}_\ell (p) \Bigr) \otimes \mathbb{1} \; , 
	\end{align*}
and 
	\begin{align*}
		\mathbb{H}_{Coul.}(\kappa)  
		& = e^{2}  \int_{V \times V} {\rm d}^3x  {\rm d}^3y \,   {\cal V}_\kappa(x-y) 
		 \mathbb{\Psi}_\kappa^{*}(x) \mathbb{\Psi}_\kappa(x)  
		\mathbb{\Psi}_\kappa^{*}(y) \mathbb{\Psi}_\kappa(y)  \otimes \mathbb{1} \; , 
	\end{align*}
act trivially on the bosonic Fock space $\mathcal{H}_{bos.}$. 

\goodbreak
Inspecting the parts of $\mathbb{H}(\kappa) $,  which involve bosonic creation and annihilation 
operators, \emph{i.e.}, 
	\begin{align}
		H^\circ_{bos.}(\kappa) & = \sum_{k \in \Gamma_\kappa} \sum_{\jmath \in \{ 1, 2 \}}  
		|k| \; a_\jmath^* (k) a_\jmath (k) \; , 
		\nonumber
		\\
		\mathbb{H}_{\, int}(\kappa) & = - e \int_{V} {\rm d}^3x \; 
		\mathbb{J}_\kappa (x) \cdot 
		\underbrace{  \sum_{k \in \Gamma_\kappa \setminus \{ 0 \}
		\atop \jmath  \in \{1,2\} } 
		\frac{\left( a_{\jmath} (k) {\rm e}^{i k \cdot x} + a_{\jmath}^* (k) {\rm e}^{i k \cdot x} \right)
		\epsilon_{\jmath} (k)}
		{\sqrt{2 V\, |k|}}  \; 
		 }_{ = A_\kappa(x)} \, ,  
		\nonumber
	\end{align}
one may try to ``complete the square'' by adding and subtracting a term of the form 
\color{black}
	\[
		 \mathbb{H}_{Curr.}(\kappa)=  
		e^{2} \sum_{k\in\Gamma_\kappa \setminus \{ 0 \}} \sum_{\jmath\in \{ 1, 2 \}} 
		\frac{\bigl( \widetilde{\mathbb{J}_\kappa}(k) \cdot \epsilon_{\jmath}(k)\bigr) 
		\bigl( \widetilde{\mathbb{J}_\kappa}(-k) \cdot \epsilon_{\jmath}(k)\bigr)}{2|k|^{2}} \; , 
	\]
which is quadratic in $\mathbb{J}_\kappa (x) $. In fact, there is a unitary operator, namely
	\begin{equation*}
		\mathbb{U}_\kappa \doteq \prod_{k\in \Gamma_\kappa \setminus \{ 0 \} \atop \jmath \in \{1, 2 \} }
		\exp \left( 
		- \frac{e \left( \bigl( \widetilde{\mathbb{J}_\kappa}(k) \cdot \epsilon_{\jmath}(k)\bigr) a_{\jmath}^{*} (k)
		- \bigl( \widetilde{\mathbb{J}_\kappa}(k) \cdot \epsilon_{\jmath}(k)\bigr)^* a_{\jmath} (k) \right) }
		{\sqrt{2|k|^3}}
		\right) \; ,  
	\end{equation*}
which accomplishes this task:  

\begin{proposition}
\label{prop:4.1}
As a quadratic form on ${\cal L} \otimes {\cal H}_{bos}$, 
	\begin{align}
				\mathbb{U}_\kappa^{*} \, \bigl( H^\circ_{bos.}(\kappa) +	\mathbb{H}_{\, int}(\kappa) \bigr) \, \mathbb{U}_\kappa 
				& =  \underbrace{ \sum_{\jmath\in \{ 1, 2 \}} \sum_{k \in \Gamma_\kappa} 
			|k| \; \mathbb{a}_{\jmath}^{*} (k) 
			\mathbb{a}_{\jmath} (k) \color{blue} }_{ := \mathbb{H}_{bos.}^{mod.} \ge 0 }  - 
		\mathbb{H}_{Curr.}(\kappa) \; ,
	\label{(38)}
	\end{align}
where the positive first term on the r.h.s.~is build up from \emph{dressed} bosonic creation and annihilation operators  
	\[ 
		\mathbb{a}_{\jmath}^* (k) =   \mathbb{1} \otimes  a_{\jmath}^* (k) 
		+ \frac{\widetilde{\mathbb{J}_\kappa}(k) \cdot \epsilon_{\jmath}(k)}{\sqrt{2|k|^3}} 
		 \otimes \mathbb{1}  
	\]
and
	\[
		\mathbb{a}_{\jmath} (k) =  \mathbb{1} \otimes  a_{\jmath} (k)
		+ \frac{\widetilde{\mathbb{J}_\kappa}(-k) \cdot \epsilon_{\jmath}(k)}{\sqrt{2|k|^3}} 
		\otimes \mathbb{1}  \; .  
 	\]
Moreover, 
	\begin{equation}
	\label{new-22}
		\mathbb{U}_\kappa^{*} \mathbb{H}_{Coul.}(\kappa)   \mathbb{U}_\kappa =  \mathbb{H}_{Coul.}(\kappa) \; .
	\end{equation}
\end{proposition}

\begin{remark}
Note that the action of $\mathbb{a}_{\jmath}^* (k)$ and $\mathbb{a}_{\jmath} (k)$ on $\mathcal{H}_{ferm.}$ is no longer trivial, as 
$U_\kappa$ mixes the components in the tensor product $\mathcal{H}_{ferm.} \otimes \mathcal{H}_{bos.}$. 
\end{remark}

\begin{proof} 
The $-e \widetilde{\mathbb{J}_\kappa}(k) \cdot \epsilon_{\jmath}(k)$ are Grassmann variables, 
which commute with all the other Grassmann variables. Thus, \eqref{(38)} 
is a consequence of 
	\begin{align*}
		\mathbb{U}_\kappa^{*} \, a_{\jmath} (k) \, U_\kappa & = 
		\mathbb{a}_{\jmath} (k) \; , 
		\qquad
		\mathbb{U}_\kappa^{*} \,  a_{\jmath}^{*} (k) \, U_\kappa = 
		\mathbb{a}^*_{\jmath} (k) \; . 
	\end{align*}
The final statement follows from the fact that on $\mathcal{L}$, the Grassmann variables commute. 
Hence, $\bigl[ \mathbb{H}_{Coul.}(\kappa) \, , \, \widetilde{\mathbb{J}_\kappa}(k) \bigr] = 0$.
\end{proof}

In the original representation on $\mathcal{H}_{ferm.}\otimes \mathcal{H}_{bos.}$ 
the current-current interaction $\mathbb{H}_{Curr.}(\kappa)$ is represented (using the 
correspondence \eqref{(29.1)}--\eqref{(29.2)}) by the normal product ${:}H_{Curr.}(\kappa){:}$ of 
	\begin{equation}
		H_{Curr.}(\kappa) \doteq 
		 \sum_{\jmath\in \{ 1, 2 \}} 
		\frac{e^{2}}{2|k|^{2}} \bigl( \widetilde{J_\kappa} (k) \cdot \epsilon_{\jmath}(k) \bigr)
		\bigl( \widetilde{J_\kappa}(-k) \cdot \epsilon_{\jmath}(k) \bigr) \otimes \mathbb{1} \; .
		\label{(39.2)}
	\end{equation}
Recalling \eqref{(3a)}, we conclude that the difference
	\begin{equation}
	\label{23}
		H_{Curr.}^{trunc.}(\kappa) :=
		H_{Curr.}(\kappa) -  {:}H_{Curr.}(\kappa){:} \;  , 
	\end{equation}
\color{black}
consists of a sum of  four  terms 
	\begin{equation}
		\label{24}
		H_{Curr.}^{trunc.}(\kappa) = \sum_{i=1}^{4} 
		 \underbrace{ \left( 
		\sum_{p,k \in \Gamma_\kappa \setminus \{ 0 \}} 
		\sum_{\ell,\ell', \ell'' \in \{ 1, 2 \} } c^{(i)}_{\ell,\ell', \ell''}(p, k) E^{(i)}_{\ell,\ell', \ell''}(p, k) \right)}_{= E_i} \; , 
	\end{equation}
obtained by  replacing  
terms of the form  $b_\ell (p) b_\ell^* (p)$ and $d_\ell (p) d_\ell^* (p)$ occurring in $H_{Curr.}(\kappa)$ by the identity
operator $\mathbb{1}$: let $\mathcal{E} (k)$ denote the $4\times 4 $ matrix
	\[
		\mathcal{E} (k) =  \frac{1}{2V\, |k|^{2} } \sum_{\jmath  \in \{ 1, 2 \} } \sum_{i=1}^3 \epsilon^{(i)}_\jmath (k) \alpha_i \; , 
	\]
then
	\begin{align*}
		E^{(1)}_{\ell,\ell', \ell''}(p, k) &=  	  b_{\ell'}^* (p+k) b_{\ell''} (p-k) \otimes \mathbb{1} \; , 
		\\
		E^{(2)}_{\ell,\ell', \ell''}(p, k) & =      b_{\ell'}^* (p+k) d^*_{\ell''} (k-p)  \otimes \mathbb{1} \; , 
		\\
		E^{(3)}_{\ell,\ell', \ell''}(p, k) &=    d_{\ell'} (-p-k)  b_{\ell''} (p-k) \otimes \mathbb{1} \; ,
		 \\ 
		E^{(4)}_{\ell,\ell', \ell''}(p, k)& =    d_{\ell''}^* (-p+k) d_{\ell'} (-p+k) \otimes \mathbb{1} \; ,
	\end{align*}
with coefficients
	\begin{align}
		c^{(1)}_{\ell,\ell', \ell''}(p, k) &=   
		\langle u_{\ell'} (p+k) , \mathcal{E} (k) u_{\ell} (p) \rangle  
		\langle u_{\ell} (p) , \mathcal{E} (k) u_{\ell''} (p-k) \rangle  \; ,  
		\nonumber \\
		c^{(2)}_{\ell,\ell', \ell''}(p, k) & =  
				 \langle u_{\ell'} (p+k) , \mathcal{E} (k)  u_{\ell} (p) \rangle 
				\langle u_{\ell} (p) , \mathcal{E} (k) v_{\ell''} (k-p)  \rangle \; , 
		\nonumber \\
		c^{(3)}_{\ell,\ell', \ell''}(p, k) &=   
				 \langle u_{\ell} (p),  \mathcal{E} (k) u_{\ell''} (p-k) \rangle 
				  \langle v_{\ell'} (-p-k) , \mathcal{E} (k) u_{\ell} (p)\rangle \; , 
		\nonumber \\
		c^{(4)}_{\ell,\ell', \ell''}(p, k)& =   
		 \langle v_{\ell'} (-p-k) , \mathcal{E} (k) u_{\ell} (p)  \rangle 
		\langle u_{\ell} (p) , \mathcal{E} (k) v_{\ell''} (-p+k) \rangle   \; .
		\label{coeff}
	\end{align}
Note that $E_{1} = E_{1}^{*}$, $E_{3} = E_{2}^{*}$, 
$E_{4}  = E_{4}^{*}$.

\bigskip
\goodbreak
An interesting aspect of the explicit form of 
$H_{Curr.}(\kappa)$ given in \eqref{(39.2)} is that is shows a certain similarity to the Coulomb interaction 
$H_{Coul.}(\kappa)$ introduced in \eqref{(18)}. We can explore this fact, by establishing an 
inequality which generalizes the positivity bound for the Coulomb energy: 

\begin{proposition}
\label{thm:4.1} 
Let $M(k)$ be the $\mathbb{C}^{4} \otimes \mathbb{C}^{4}$-valued matrix given by
	\begin{equation}
		M(k) \equiv \mathbb{1} \otimes \mathbb{1}  
		-\frac{1}{2} \alpha \epsilon_{1}(k) \otimes \alpha \epsilon_{1}(k)
		- \frac{1}{2} \alpha \epsilon_{2}(k) \otimes \alpha \epsilon_{2}(k)
	\label{(43a)}
	\end{equation}
and let $\mathcal{V}$ be (up to a constant $\tfrac{1}{V}$) be 
the Coulomb kernel introduced in \eqref{(19)}. Then
	\begin{equation}
		 \int_{V \times V} {\rm d}^3 x {\rm d}^3 y \; 
		\mathcal{V}(x-y) \Bigl( (\Psi(x)^{*}\otimes \Psi(y)^{*}), M(k)(\Psi(x) \otimes \Psi(y)) \Bigr) \ge 0
	\label{(43)}
	\end{equation}
in the sense of quadratic forms on $\mathcal{H}_{ferm.}$. 
Note that the $\mathbb{C}^{4}$-valued bounded operators
$\Psi(x)$ and $\Psi^{*}(x)$ were defined in \eqref{(11)}. 
\end{proposition}

\begin{proof} 
We may write $M(k)$ in the form
	\begin{align*}
		4M(k)& =\Bigl( (\mathbb{1}-\alpha \epsilon_{1}(k)) \otimes(\mathbb{1}+\alpha \epsilon_{1}(k))
		+(\mathbb{1}+\alpha  \epsilon_{1}(k)) \otimes(\mathbb{1}-\alpha  \epsilon_{1}(k)) \Bigr)
		\nonumber \\
			& 
		\quad + \Bigl((\mathbb{1}-\alpha  \epsilon_{2}(k)) \otimes(\mathbb{1}+\alpha  \epsilon_{2}(k))
		+(\mathbb{1}+\alpha \epsilon_{2}(k)) \otimes(\mathbb{1}-\alpha \epsilon_{2}(k)) \Bigr) \; . 
	\end{align*}
By a unitary transformations on can convert each of the four summands in $4M(k)$ to a 
direct product of diagonal matrices $D_{1,i} \otimes D_{2,i}, i=1, \cdots, 4$. 
The eigenvalues of each $D_{1,i},D_{2,i}, i=1, \cdots, 4$ are zero and two, each with 
multiplicity two, as may be verified by straightforward diagonalisation, independent of $k$, 
because they depend only on the Euclidean norms $|\epsilon_{1}(k)|$ and $|\epsilon_{2}(k)|$, 
which are both equal to one. The inequality~\eqref{(43)} then follows.  
\end{proof}

\begin{corollary}
\label{Corollary 4.1}
$H_{CC} \doteq H_{Coul.}(\kappa) - H_{Curr.}(\kappa) \ge 0$.
\end{corollary}

\begin{proof} Inspecting \eqref{(43a)}, we see that $H_{Coul.}(\kappa)$ as defined in \eqref{(18)} and 
$H_{Curr.}(\kappa) $ defined in \eqref{(39.2)} 
can be combined to yield the expression on the l.h.s.~in \eqref{(43)}. 
\end{proof}

We will also need a bound on general Fermi operators.  
As before, the Fermi creation operators $c^*_i$ and annihilation operators $c_j$ 
will satisfy
canonical anti-commutation relations, \emph{i.e.},  $\{c_i,c_k^{*}\} = \delta_{i,k}$ 
and the remaining anti-commutators are all equal to zero. 

\begin{lemma}[Bogoliubov \&~Bogoliubov \cite{BoBo}] Let $F$ be a self-adjoint fermionic Hamiltonian of the general form
	\begin{equation}
		F  
		= \sum_{i,j=1}^n A_{i,j} c_i^{*} c_j +\frac{1}{2}\sum_{i,j=1}^n B_{i,j} c_i^{*}c_j^{*}  
		+ \frac{1}{2} \sum_{i,j=1}^n B_{i,j}^{*} c_i c_j \; , 
	\label{(48)}
	\end{equation}
where the matrix~$A= (A_{i,j})$ is hermitian  
and the matrix $B= (B_{i,j})$ is anti-symmetric, \emph{i.e.}, 
	\[
		A=A^* \; , \qquad  B  = -B^{T} \; , 
		\qquad  \text{(\,$T$ denotes the transpose)} \; . 
	\]
Then there exists a constant $\mu$ such that 
$F + \mu  \cdot \mathbb{1} \ge 0$.
\end{lemma}

\begin{proof}

It has been shown in \cite{BoBo} that a Hamiltonian of the form \eqref{(48)} can always be written as 
	\begin{equation}
		F
		= \sum_{\ell = 1}^{n} \lambda_\ell \, q^{*}_\ell q_\ell + \mu \cdot \mathbb{1} \; , 
		\qquad \lambda_\ell \ge 0 \; , 
	\label{(56)}
	\end{equation}
with new Fermi creation and annihilation operators
	\begin{equation}
	\label{U-V}
		q^{*}_\ell  = \sum^n_{ i=  1} 
		\bigl( \overline{U_{i, \ell}}  \, c_{i}^{*} + 
		\overline{W_{i, \ell}} \, c_{i} \bigr) \; , 
	\quad
		q_\ell  = \sum^n_{ i=  1} 
		\bigl(U_{i, \ell}  \, c_{i} + W_{i, \ell}  \, c_{i}^{*} \bigr) \; , 
	\end{equation}
and the new Fermi operators $q_\ell , q^*_{\ell'}$ satisfying canonical anti-commutation relations. 
The matrices $U = (u_{i, \ell}) $ and $W=(v_{i,\ell}) $ appearing in \eqref{U-V}
satisfy 
	\begin{equation}
		A U + B W^{*} = U \Lambda \; , 
		\qquad
		-B^{*} U - A^{*} W^{*} = W^{*} \Lambda \; ,  
	\label{(58)}
	\end{equation}
and $\mu$ is given by $\mu = - \sum_{i=1}^{n} \lambda_{i} \operatorname{Tr} W^* W  $ and $\Lambda$ denotes
the diagonal $n \times n$ matrix of the $\{\lambda_{j}\}$.
\end{proof}

\begin{theorem}
\label{thm:5.2} 
There exists a real number  $\infty> \mu(\kappa)>0$ (the vacuum energy renormalization)
such that the first renormalized Hamiltonian 
	\[
		H_{1, ren.}(\kappa) \equiv H(\kappa) + \mu(\kappa) \cdot \mathbb{1} \; , 
	\]
where $H(\kappa)$, 
the original Hamiltonian describing qed in the
Coulomb gauge, introduced in \eqref{(1)}, satisfies
	\begin{equation}
		H _{1, ren.}(\kappa) \ge 0 \; , 
		\label{(64)}
	\end{equation}
as a quadratic form on ${\mathcal H}_{ferm.} \otimes {\mathcal H}_{bos.}$.
\end{theorem}

\begin{proof} It follows from \eqref{(33.1)}, \eqref{(38)} and \eqref{new-22} that 
	\begin{equation}
		\mathbb{H}(\kappa) \ge \mathbb{H}^\circ_{ferm.}(V) + \mathbb{H}_{Coul.}(\kappa)  
			- \mathbb{H}_{Curr.}(\kappa) \; , 
		\label{(*)}
	\end{equation}
as operators on ${\cal L} \otimes {\cal H}_{bos}$.  
Moreover, by the correspondence \eqref{(29.1)}--\eqref{(29.2)}, $\mathbb{H}_{Coul.}(\kappa)$ 
is represented on ${\cal H}_{ferm}$ by

	\[
		  {:}H_{Coul.}(\kappa){:} = H_{Coul.}(\kappa) 
		- \sum_{i=1}^{4} E^{tr.}_i   \; , 
	\]
where $E^{trunc.}_i $ arises by replacing the term 
$\frac{1}{2} \sum_{\jmath} \sum_i \alpha_i \epsilon^{(i)}_{\jmath}( k)$ in the explicit expressions for the $E_i$'s  
in \label{E-i} appearing in \label{(24)} by the identity. 

\goodbreak
Coming back to the original representation on ${\cal H}_{ferm} \otimes {\cal H}_{bos}$, taking into account that the operator
on the r.h.s.~of \eqref{(*)} acts trivially on ${\cal H}_{bos}$, and using \eqref{23}, \eqref{24} and \eqref{coeff}, 
\eqref{(*)} and the above-mentioned representation
for $\mathbb{H}_{Coul.}(\kappa)$, it follows that
        \begin{align}
		 H(\kappa) & \ge  H^\circ_{ferm.}(V)
		- \underbrace{ \bigl( H_{Curr.}(\kappa) - \sum_{i=1}^{4} E_i(\kappa) \bigr)}_{ = \, {:} \, H_{Curr.}(\kappa) \,{:} }
		 \nonumber \\
		 & \qquad \qquad \qquad \qquad \qquad \qquad \qquad \qquad  
		 + 
		\underbrace{H_{Coul.}(\kappa) -  \sum_{i=1}^{4} E^{tr.}_i (\kappa)}_{ = \, {:} \, H_{Coul.}(\kappa) \,{:} } 
		\nonumber \\
		& \ge 	\underbrace{ H^\circ_{ferm.}(V)+ 
		\sum_{i=1}^{4} \bigl( E_i (\kappa)
		-  E^{tr.}_i (\kappa) \bigr) }_{ =: H^{mod.}_{ferm.}(\kappa)
		\ge - \mu (\kappa) \cdot \mathbb{1} }
		+  \underbrace{ H_{Coul.}(\kappa) - H_{Curr.}(\kappa)}_{= H_{CC} \ge 0 } 
		\nonumber \\
		& \ge  - \mu (\kappa) \cdot \mathbb{1} \; . 
		\label{h-k-b}
	\end{align}
\color{black}
In the last inequality we have used Corollary \ref{Corollary 4.1}
and the fact that according to its definition given in~\eqref{h-k-b},  
the Fermi operator $H^{mod.}_{ferm.} (\kappa)$ is 
of the general form described in~\eqref{(48)}. 
(Note that we have replaced $H^\circ_{ferm.} (V)$ by $H^\circ_{ferm.} (\kappa)$; the latter is 
the free fermion Hamiltonian  with the 
restriction to $k \in \Gamma_\kappa$.). 
\end{proof} 

\begin{remark}
The bound holds for all finite $V$ and $\Lambda$.  
In particular, if $E(\kappa)$ denotes the lowest eigenvalue of~$H (\kappa) + \mu(\kappa)$, then the lowest
accumulation point  of the sequence $\{ E(\kappa) \}$ is positive.
\end{remark}

\begin{remark}
As remarked above, the form of the representation of $\mathbb{H}_{Coul.}(\kappa)$ on ${\cal H}_{ferm}$ 
is simply due to the Wick dots in \eqref{(18)}. This choice is, on the other hand, motivated by
the fact that in perturbation theory there are cancellations between the instantaneous Coulomb interaction \eqref{(18)} and the transverse
term \eqref{(3)}, leading to a final covariant propagator (\cite{JJS}, pp. 252--253). Since ${:}H_{Curr}(\kappa){:}$ is the operator which arises
from the transverse term as a consequence of the first unitary transformation, it seems advisable to define the instantaneous
Coulomb interaction in an analogous way, in order that a covariant propagator may 
arise in the limits $\Lambda \to \infty$, followed by $V \to \infty$. 
\end{remark}

\section{The issue of triviality}

There are two different reasons why the present theory might be trivial. 

One is the infrared problem: since the photons have been \emph{decoupled} by applying $U_\kappa$, 
the sums over $k \in \Gamma_{\kappa}\setminus \{ 0 \}$ in $H_{CC}$ 
and $H^{mod.}_{ferm.}(\kappa)$ in \eqref{(48)} might lead to divergence (to $+\infty$) in the matrix 
elements of $H_{CC}$ and $H^{mod.}_{ferm.}(\kappa)$ (as $V \to \infty$). 
Since $H_{CC}:= H_{Coul.}(\kappa) - H_{Curr.}(\kappa)$ is equal to the l.h.s.~in \eqref{(43)},
and together with the fact that
	\begin{equation}
		\sup_{V} \frac{1}{V} \sum_{k \in (B \cap \Gamma_{\kappa} ) \setminus \{ 0 \} } 
		\frac{1}{|k|^{2}} \le c < \infty \; , 
		 \label{(68)}
	\end{equation}
where $B$ is some fixed ball centered at the origin, and $c$ a constant independent of $V$, 
$H_{CC}$ is seen to have good infrared behavior. 

For $H^{mod.}_{ferm.}(\kappa)$ 
one should examine $E_i$, $i=1, \ldots, 4$ in \eqref{(24)}. This is seen to involve a point-wise 
bound on the polarization vectors \eqref{(14.3)}. This requires some manipulations (we remind the 
reader that we follow the notation of Sakurai's book \cite{JJS}). We have
	\begin{equation}	
		\sum_{\sigma} v_{p,\sigma} v_{p,\sigma}^{+} = -\frac{(i \gamma \cdot p+m)}{2m} \; , 
		\quad
		\sum_{\sigma} u_{p,\sigma} u_{p,\sigma}^{+} = \frac{(-i \gamma \cdot p+m)}{2m} \; .   
	\label{(69)}
	\end{equation}
Using \eqref{(69)}, we may write the matrix summands in 
\eqref{coeff}, which involve the polarization vectors, in the typical 
form (up to some changes of sign and symbols):
	\begin{align}
		(\gamma \cdot \epsilon_{k,\alpha}) (\gamma \cdot p_{1}) (\gamma \cdot \epsilon_{k,\alpha})  
		& = \gamma \cdot \epsilon_{k,\alpha} (2p_{1} \cdot \epsilon_{k,\alpha}) - 
		(\gamma \cdot \epsilon_{k,\alpha})(\gamma \cdot p_{1}) 
		\nonumber \\
		& = 2p_{1} \cdot  \epsilon_{k,\alpha} \gamma \cdot \epsilon_{k,\alpha} - \gamma \cdot p_{1} \; .
	\label{(69.3)}
	\end{align} 
We now denote the absolute value of a number $a$ by $\| a \|$, and the components of a three-vector $v$ by $v_{i}, i=1, \cdots, 3$.
We have 
	\begin{equation}
		\| 2p_{1} \cdot \epsilon_{k,\alpha} \| \le 2 \, | p_{1}| \, |\epsilon_{k,\alpha}| = 2 \, |p_{1}| 
		\label{(69.4)} \; . 
	\end{equation}
Recall that the dot denotes the scalar product in Euclidean three-space, and  $| \, . \, |$ the Euclidean norm. 
By \eqref{(69.3)} and \eqref{(69.4)}, we are therefore reduced to finding a bound
	\begin{equation}
		\sup_{V,\alpha} \frac{1}{V} \sum_{k \in (B \cap \Gamma_{\kappa})\setminus \{ 0 \} } 
		\frac{1}{|k|^{2}} \max_{i=[1,3]}||\epsilon_{k,\alpha,i}|| \le c < \infty \, , 
	\label{(69.5)}
	\end{equation}
again for some fixed ball centred at the origin. Verification of \eqref{(69.5)} is straightforward, using \eqref{(14.3)}. 

\goodbreak

A second possible reason for triviality is \emph{charge renormalization}. Since the cutoff on 
the fermions should be independent of that in the photon field (as, in fact, adopted in \cite{JLW}). 
We now set a different cutoff $\Lambda^{'}$ on the fermion field \eqref{(11)} 
for the purpose of discussion.  Examination of $H_{CC}$ shows that, setting
	\begin{equation}
		e^{2} \equiv e_{bare}^{2} = O \left( |\Lambda|^{(-1-\epsilon)/3} \right) e_{ren}^{2}
	\label{(70)}
	\end{equation}
(recall that $|\Lambda|$ denotes the number of sites in the set $\Lambda$, 
which for a cube is of the same order as of a sphere of radius $\Lambda^{1/3}$ ), we would obtain
	\[
		\lim_{\Lambda \to \infty} H_{CC} = 0 \; , 
	\]
and the coefficients in \eqref{coeff}  show that
	\[	
		\lim_{\Lambda \to \infty}  
		\sum_{i=1}^{4} \bigl( E_i (\kappa)
		-  E^{tr.}_i (\kappa) \bigr) = 0 
	\]   
(note that $E_i (\kappa)$ and $E^{tr.}_i (\kappa)$ both depend on $\Lambda$, which was not explicitly indicated), 
under the same condition \eqref{(70)}. 
\color{black}
Thus, a suitable ``charge renormalization'' trivializes the 
theory, in the sense that the Hamiltonian reduces to the \emph{free} Hamiltonian, when the 
ultraviolet cutoff is removed (for fixed $m=m_{bare}$). The well known charge renormalization 
in qed (see \cite[pp.~279--283]{JJS} or \cite[pp.~445, 462]{Weinb1})
	\begin{equation}
		e_{bare}^{2} = \frac{1}{Z_{3}} e_{ren}^{2} \;  , \quad \text{with} \quad
		Z_{3} = 1 - \frac{e^{2}}{12\pi^{2}} \log \frac{\Lambda^{2/3}}{m^{2}} \; , 
		\label{(73.2)}
	\end{equation}
would \emph{not}, however, trivialize the theory. We conclude that there are no a priori reasons 
leading to suspect that qed, either with bare parameters, or~\eqref{(73.2)}, is a trivial theory. We emphasize, 
however, that there is also no a priori reason why \eqref{(73.2)} must be incorporated to show global existence 
of qed, and, for this reason, we did not investigate in detail the effect of 
replacing $m=m_{bare}$ by its renormalized version \cite{JJS}, see also \cite{Frohlich} for a discussion 
of mass renormalization in non-relativistic qed and references. For instance, a non-trivial S matrix
(with no mass or charge renormalizations) would not be contradictory with the same S matrix having an
asymptotic, but divergent, asymptotic expansion in terms of the fine structure constant, the latter, however,
making sense only with the usual mass and charge renormalisations (\cite{Weinb1, JJS}).

\section{Conclusion}

Our method, using the space ${\cal L}$ of Section 3, permits the use of a unitary transformation to 
convert the Fermi Hamiltonian of $QED_{1+3}$ in the Coulomb gauge to the form \eqref{(38)}, in which fermions 
interact both through the instantaneous Coulomb force and an additional term incorporating 
the ``photon cloud''. Going back to the original fermion Fock space, the free fermion Hamiltonian 
is modified by an additional quadratic term,  
	\[
		H^{mod.}_{ferm.}(\kappa) = H^\circ_{ferm.}(V)- H^{trunc.}_{Curr.}(\kappa)-H_{Coul.}^{trunc}(\kappa) \; , 
	\]
which may be interpreted as a \emph{Bremsstrahlung} term (see \cite{JJS}, p.~229--231). Thus, the \emph{photon cloud} also 
acts to \emph{dress} the electron-positron field. This suggests that the spectrum of the physical 
Hamiltonian $H_{phys.}$   
may be purely absolutely continuous, as suggested 
by Buchholz's important result~\cite{Buch} (proven under very reasonable assumptions) that 
in $QED_{1+3}$ there exist no eigenvalues of the mass operator $M^{2} = H_{phys.}^{2}- \vec{P}_{phys.}^{2}$, 
where $\vec{P}_{phys.}$ denotes the physical momentum.  

The most significant point about the unitary operator $U_\kappa$ is that it generalizes 
well-known transformations which, in the limits $\Lambda \to \infty$, followed 
by $V \to \infty$, lead to \emph{inequivalent representations} of the canonical commutation 
relations, for multiple reasons, reviewed in A.S.Wightman's lectures in~\cite{Wight}: infrared and ultraviolet 
divergences (note that due to $\lambda_{k} = (\omega_{k}\sqrt{2\omega_{k}})^{-1}$, $ \int {\rm d}k \, |\lambda_{k}|^{2}$ has both infrared 
and ultraviolet logarithmic divergences), as well as Euclidean invariance associated 
with Haag's theorem and vacuum polarization. 
The fact that non-Fock representations are imperative in the 
theory without cutoffs explains that we have obtained an apparently more realistic physical 
picture in the cutoff version, with reasonable properties when the cutoffs are removed (Theorem~\ref{thm:5.2}).

It is possible that $\mu(\kappa) = -\rho V  + \text{(divergent terms)} $ 
as $\Lambda \to \infty$. In that case, of course, only the divergent
terms should be renormalized. $-\rho$ may be related to the positronium energy (see \cite{JJS}). An indication that $\mu = O(V)$ 
is given by the fact that it involves a sum over momentum modes 
(see Proposition~\ref{prop:4.1}), and, indeed, a divergent part of type $\mu_{\Lambda}V$
with $\mu_{\Lambda} \to \infty$ as $\Lambda \to \infty$  would be the analogue of the ``chemical potential renormalization'' mentioned
in the introduction in connection with non-relativistic qed, with $V$  replacing $N$, the number of electrons, which is not a 
good quantum number in the relativistic case. It would thus be of great importance
to investigate the unitary transformation of Proposition~\ref{prop:4.1} in greater detail.
\color{black}

As a final remark, it is important to note that the Heisenberg picture time evolution is left 
invariant by a c-number vacuum energy 
renormalization. The positive renormalized Hamiltonian obtained in Theorem \ref{thm:5.2} is therefore the correct Hamiltonian in the
sense of automorphisms of the algebra of observables, in the proper limits.

\color{black}

\end{document}